\documentclass[11pt]{article}
\usepackage{amsmath,amssymb,amsbsy,amsfonts,amsthm,latexsym,
               amsopn,amstext,amsxtra,euscript,amscd,color}
\usepackage{graphicx}

\makeatletter
\renewcommand\title[1]{\gdef\@title{\reset@font\Large\bfseries #1}}
\renewcommand\section{\@startsection {section}{1}{\z@}%
                                   {-3.5ex \@plus -1ex \@minus -.2ex}%
                                   {2.3ex \@plus.2ex}%
                                   {\normalfont\large\bfseries}}
\renewcommand\subsection{\@startsection{subsection}{2}{\z@}%
                                     {-3ex\@plus -1ex \@minus -.2ex}%
                                     {1.5ex \@plus .2ex}%
                                     {\normalfont\normalsize\bfseries}}
\renewcommand\subsubsection{\@startsection{subsubsection}{3}{\z@}%
                                     {-2.5ex\@plus -1ex \@minus -.2ex}%
                                     {1.5ex \@plus .2ex}%
                                     {\normalfont\normalsize\bfseries}}

\def\@runningauthor{}\newcommand{\runningauthor}[1]{\def\runningauthor{#1}}
\def\@runningtitle{}\newcommand{\runningtitle}[1]{\def\runningtitle{#1}}

\renewcommand{\ps@plain}{%
\renewcommand{\@evenhead}{\footnotesize\scshape \hfill\runningauthor\hfill}
\renewcommand{\@oddhead}{\footnotesize\scshape \hfill\runningtitle\hfill}}
\g@addto@macro\bfseries{\boldmath}

\makeatother

\usepackage{amsthm,amsmath,amssymb,enumerate,multirow}

\usepackage{graphicx}

\usepackage[colorlinks=true,citecolor=black,linkcolor=black,urlcolor=blue]{hyperref}

\usepackage{authblk}

\theoremstyle{plain}
\newtheorem{theorem}{Theorem}[section]
\newtheorem{lemma}[theorem]{Lemma}
\newtheorem{corollary}[theorem]{Corollary}
\newtheorem{proposition}[theorem]{Proposition}

\theoremstyle{definition}
\newtheorem{definition}[theorem]{Definition}
\newtheorem{example}[theorem]{Example}
\newtheorem{conjecture}[theorem]{Conjecture}

\theoremstyle{remark}
\newtheorem{remark}[theorem]{Remark}

\title{Partially APN  Boolean functions and classes of functions that are not APN infinitely often}

\runningtitle{Partially APN  Boolean functions}

\author[1]{Lilya Budaghyan}
\author[1]{Nikolay S. Kaleyski}
\author[2]{Soonhak Kwon}
\author[3]{Constanza Riera}
\author[4]{Pantelimon St\u anic\u a}
\affil[1]{\small Department of Informatics, University of Bergen, \protect\\ 5020 Bergen, Norway;
\{\tt Lilya.Budaghyan, Nikolay.Kaleyski\}@uib.no}
\affil[2]{\small Department of Mathematics, Sungkyunkwan University,\protect\\ Suwon 16419,  Republic of Korea; {\tt shkwon@skku.edu}}
\affil[3]{Department of Computing, Mathematics, and Physics,\protect\\ Western Norway University of Applied Sciences, \protect\\ 5020 Bergen, Norway; {\tt csr@hvl.no} }
\affil[4]{Department of Applied Mathematics, Naval Postgraduate School,\protect\\ Monterey, CA 93943-5212, U.S.A.; {\tt pstanica@nps.edu} }

\runningauthor{L. Budaghyan, N. S. Kaleyski, S. Kwon, C. Riera, P. St\u anic\u a}

\date{}

\def \F {{\mathbb F}}

\def\cW{{\mathcal H}}

\def\cW{{\mathcal W}}

\def\00{{\bf 0}}
\def\11{{\bf 1}}
\def\+{\oplus}

\def \F {{\mathbb F}}

\def \Tr {{\rm Tr}_1^n}

\newcommand{\cardinality}[1]{\# #1}

\newcommand{\bwht}[2]{\mathcal{W}_{#1}(#2)}
\newcommand{\vwht}[3]{\mathcal{W}_{#1}(#2,#3)}
\newcommand{\vwhtp}[4]{\mathcal{W}^{#4}_{#1}(#2,#3)}

\begin{document}

\maketitle

\thispagestyle{empty}

\begin{abstract}
In this paper we define a notion of partial APNness  and find various characterizations and constructions of classes of functions satisfying this condition. We connect this notion to the known conjecture that APN functions modified at a point cannot remain APN. In the second part of the paper, we find  conditions for some transformations not to be partially APN, and in the process, we find classes of functions that are never APN for infinitely many extensions of the prime field $\F_2$, extending some earlier results of Leander and Rodier.
\end{abstract}




\section{Introduction}

The objects of this study are Boolean functions and some of their differential properties.  We will introduce here only some needed notions, and the reader can consult~\cite{Bud14,CH1,CH2,CS17} for more on Boolean functions.

Let $n$ be a positive integer and $\F_{2^n}$ denote the  finite field with $2^n$ elements, and $\F_{2^n}^*=\F_{2^n}\setminus\{0\}$. Further, let $\F_2^m$ denote the $m$-dimensional vector space over $\F_2$.
We call a function from $\F_{2^n}$ to $\F_2$  a {\em Boolean function} on $n$ variables.
For $f:\F_{2^n}\to \F_2$ we define the {\it Walsh-Hadamard transform} to be the integer-valued function
$\displaystyle
\bwht{f}{u}  = \sum_{x\in \F_{2^n}}(-1)^{f(x)+\Tr(u x)}, \ u \in \mathbb{F}_{2^n},
$
 where $\Tr:\F_{2^n}\to \F_2$ is the absolute trace function, given by $\Tr(x)=\sum_{i=0}^{n-1} x^{2^i}$.
This transform satisfies Parseval's relation
$\displaystyle
\sum_{a \in \F_{2^n}} \bwht{f}{a}^2 = 2^{2n}.
$


Given a Boolean function $f$, the derivative of $f$ with respect to~$a \in \F_{2^n}$ is the Boolean function
$
 D_{a}f(x) =  f(x + a)+ f(x), \mbox{ for  all }  x \in \F_{2^n}.
$

For positive integers $n$ and $m$, any map $F:\F_2^n\to\F_2^m$ is called a vectorial Boolean function, or $(n,m)$-function. When $m=n$, $F$ can be uniquely represented as a univariate polynomial over $\F_{2^n}$ (using the natural identification of the finite field with the vector space) of the form
$
F(x)=\sum_{i=0}^{2^n-1} a_i x^i,\ a_i\in\F_{2^n}.
$
The algebraic degree of $F$ is then the largest Hamming weight of the exponents $i$ with $a_i\neq 0$. For an $(n,m)$-function $F$, we define the Walsh transform $\vwht{F}{a}{b}$ to be the Walsh-Hadamard transform of its component function ${\rm Tr}_1^m(bF(x))$ at $a$, that is,
\[
  \vwht{F}{a}{b}=\sum_{x\in\F_{2^n}} (-1)^{{\rm Tr}_1^m(bF(x))+\Tr(ax)}, \text{ where $a\in \F_{2^n}, b\in \F_{2^m}$.}
\]

For an $(n,n)$-function $F$, and $a,b\in\F_{2^n}$, we let $\Delta_F(a,b)=\cardinality{\{x\in\F_{2^n} : F(x+a)+F(x)=b\}}$, where $\cardinality{S}$ denotes the cardinality of a set $S$. We call the quantity
$\Delta_F=\max\{\Delta_F(a,b)\,:\, a,b\in \F_{2^n}, a\neq 0 \}$ the {\em differential uniformity} of $F$. If $\Delta_F\leq \delta$, then we say that $F$ is differentially $\delta$-uniform. If $\delta=2$, then $F$ is called an {\em almost perfect nonlinear} ({\em APN}) function.
There are many useful characterizations and properties of APN functions, some of which are stated below (see~\cite{CB18,CH2,CV95,Rod09}).
\begin{lemma}
\label{APN:char1}
Let $F$ be an $(n,n)$-function. The following hold$:$
\begin{enumerate}[$(i)$]
\item we have
$\displaystyle
\sum_{a,b\in \F_{2^n}} \vwhtp{F}{a}{b}{4} \geq 2^{3n+1}(3\cdot 2^{n-1}-1),
$
with equality if and only if  $F$ is APN$;$
\item if $F(0)=0$ and $F$ is APN, then
$\displaystyle
\sum_{a,b\in \F_{2^n}} \vwhtp{F}{a}{b}{3}=2^{2n+1}(3\cdot 2^{n-1}-1);
$
\item  $($Rodier Condition$)$ $F$ is APN if and only if all the points $x,y,z$ satisfying
$
F(x)+F(y)+F(z)+F(x+y+z)=0,
$
belong to    the curve $(x+y)(x+z)(y+z)=0$.
\end{enumerate}
\end{lemma}


We next introduce the notion of a partial APN function.
\begin{definition}
Let $x_0\in\F_{2^n}$. We call an $(n,n)$-function $F$ a ({\em partial}) {\em $x_0$-APN function}, or simply $x_0$-APN function, if all the points $u,v$ satisfying
$
F(x_0)+F(u)+F(v)+F(x_0+u+v)=0,
$
belong to the curve $(x_0+u)(x_0+v)(u+v)=0$. 
\end{definition}

Alternatively, we can say that a function $F$ is $x_0$-APN if for any $a \ne 0$ the equation $F(x+a) + F(x) = F(x_0 + a) + F(x_0)$ has only two solutions. Certainly, an APN function is an $x_0$-APN for any point $x_0$.

 A function $F$ is called {\em weakly APN} if for any $a\ne0$ the function $F(x+a)+F(x)$ takes at least $2^{n-2}+1$ different values (see \cite{WeaklyAPN}).  Note that the notion of partial APN function  differs from the notion of weakly APN function.
 For example, it can be checked that $F(x)=x^{2^n-2}$ over $\F_{2^n}$ with $n$ even is weakly APN but not $x_0$-APN, for $x_0\in\F_{2^n}$. On the other hand, $F(x) = x^7$ over $\F_{2^{11}}$ is $0$-APN but not weakly APN.

Our proposal for the partial APN concept comes from a study of the conjecture in~\cite{CB18}, which claims that for  $n\ge3$ an APN function modified at a point cannot remain APN. While the start of this work has some initial study overlap with~\cite{CB18}, our ultimate goal is to investigate the partial APN concept.

Our paper is organized as follows. In Section~\ref{sec2} we introduce the one point modification of an $(n,n)$-function and investigate its Walsh coefficients' third and fourth moments  as compared to the original function. We further give a (local-global principle) characterization for the APNess of the modified version of an APN function, which was the original starting point of this investigation. A conjecture is proposed here, slightly strengthening the original conjecture of~\cite{CB18}.  Section~\ref{sec3} contains a standalone characterization of the partial APN concept in terms of the third moments. In  Section~\ref{sec4} we continue with some constructions and characterization of the pAPN property for monomial functions (in particular, we show that for power functions, the pAPN at a nonzero point will imply APNess, and, in general, the pAPNess at a nonzero point will imply APNess for quadratic functions). In  Section~\ref{sec5}, in the spirit of Rodier et al. we concentrate on the various linear transformations of some functions to show (non)pAPNess and in the process we show a much stronger version of a result by Leander and Rodier~\cite{LR11}. Section~\ref{sec6} contains the conclusion and further comments.


\section{Boolean functions modified at a point}
\label{sec2}

Let $F:\F_{2^n}\to\F_{2^n}$ and consider an arbitrary point $x_0 \in \mathbb{F}_{2^n}$ and some nonzero $\epsilon \in \mathbb{F}_{2^n}^*$. Denote $y_0 = F(x_0)$ and $y_1 = y_0 + \epsilon$. Then the function $F'$ over $\mathbb{F}_{2^n}$ defined by 
\begin{equation}
\label{F'}
F'(x)=\begin{cases}
F(x) & \text{ if } x\neq x_0\\
y_1 & \text{ if } x= x_0
\end{cases}
\end{equation}
is called a ({\em single point}) {\em $(x_0,\epsilon)$-modification of $F$}.

It is rather easy to show that there are  single point modifications of an APN function $F$ that are not APN.
\begin{proposition}
  If an $(n,n)$-function $F$ is APN for $n > 1$, then for any $x_0\in\F_{2^n}$ there exists $\epsilon\in\F_{2^n}^*$ such that the $(x_0,\epsilon)$-modification of $F$ is not APN.
\end{proposition}
\begin{proof}
  Suppose $F$ is APN and $x_0 \in \F_{2^n}$ is given. Take $y,z \in \mathbb{F}_{2^n}$ such that $x_0$, $y$ and $z$ are distinct and let $F'$ be the $(x_0, \epsilon=F(y) + F(z) + F(x_0 + y + z)-F(x_0))$-modification of $F$. Then we have $F'(x_0) \ne F(x_0)$ since $F$ is APN and $F'(x_0) + F'(y) + F'(z) + F'(x_0 + y + z) = 0$ so that $F'$ cannot be APN.
\end{proof}

Next, we find some necessary and sufficient conditions for an $(x_0,\epsilon)$-modification of a given function to be partially APN.

\subsection{Preliminary lemmas}

\begin{lemma}
\label{Walsh-diff}
Let  $F$ be an $(n,n)$-function and $F'$ be an $(x_0,\epsilon)$-modification of $F$ for $x_0,y_1=y_0+\epsilon \in \mathbb{F}_{2^n}$ and $y_1 \ne y_0 = F(x_0)$. Then, 
\[
  \vwht{F'}{a}{b}=\vwht{F}{a}{b}-(-1)^{\Tr(ax_0+by_0)}(1-(-1)^{\Tr(b\,\epsilon)}).
\]
\end{lemma}
\begin{proof}
We have
\begin{align*}
\cW_{F'}(a,b)&=\sum_{x\in\F_{2^n}} (-1)^{\Tr(b F'(x)+ax)}
=\sum_{x\neq x_0} (-1)^{\Tr(b F(x)+ax)}+ (-1)^{\Tr(b y_1+ax_0)}\\
&= \sum_{x\in\F_{2^n}} (-1)^{\Tr(b F(x)+ax)}+(-1)^{\Tr(ax_0+by_1)}-(-1)^{\Tr(ax_0+by_0)},
\end{align*}
which justifies our claim.
\end{proof}
For any given elements $a,b\in\F_{2^n}$, we let
$
E_F(a,b)=(-1)^{\Tr(ax_0+by_0)}D_F(b),
$
where $D_F(b)=1-(-1)^{\Tr(b\,\epsilon)}$. Note that $E_F(a,b)$ depends on $x_0$, $y_0, y_1$.
The following lemma can be easily shown by induction.

\begin{lemma}
  \label{lemmaPowerSimplification}
  Let  $F$ be an $(n,n)$-function and let $x_0,y_1 \in \F_{2^n}$ with $y_1 \ne y_0 = F(x_0)$ and $\epsilon = y_0 + y_1$.
Then for any integer $m\geq 1$ and any elements $a,b\in\F_{2^n}$, we have
\begin{enumerate}[$(i)$]
  \item  $E_F^{2m}(a,b)=2^{2m-1}D_F(b)$, and
\item  $E_F^{2m+1}(a,b)=2^{2m}E_F(a,b)$.
\end{enumerate}
\end{lemma}

\subsection{The third and fourth moments  and an APN characterization of a one point modification of an APN function }

In the following we make use of the Kronecker function $\delta_0(z)=\begin{cases} 1 & \text{if } z=0\\
0 & \text{if } z\neq 0. \end{cases}$
\begin{theorem}
\label{4thpower}
Let $F$ be an $(n,n)$-function and $F'$ be its $(x_0,\epsilon)$-modification for some $x_0,y_1=y_0+\epsilon \in \mathbb{F}_{2^n}$ with $y_1 \ne y_0 = F(x_0)$.  Then the following hold:
\begin{enumerate}[$(i)$]
\item[$(i)$] $\displaystyle \frac14\sum_{a,b\in \F_{2^n}}\left(
  \vwhtp{F}{a}{b}{4}-\vwhtp{F'}{a}{b}{4}\right) = \sum_{a,b\in \F_{2^n}}
  \vwhtp{F}{a}{b}{3} E_F(a,b) -(3\cdot 2^{3n}-2^{2n+1});$
\item[$(ii)$] $\displaystyle \sum_{a,b\in \F_{2^n}}\left(\vwhtp{F}{a}{b}{3}-\vwhtp{F'}{a}{b}{3}\right)
  =3\sum_{a,b\in \F_{2^n}}  \vwhtp{F}{a}{b}{2} E_F(a,b)
-3\cdot 2^{2n+1}  \\ \cdot \left(\delta_0(F(0))- \delta_0(y_1-y_0+F(0))\right)
+2^{2n+2} \delta_0(x_0) \left(\delta_0(y_0)-\delta_0(y_1) \right).$
\end{enumerate}
\end{theorem}

\begin{proof}
We show $(i)$ first.
Taking fourth powers in the identity $\vwht{F'}{a}{b}=\vwht{F}{a}{b}-E_F(a,b)$ of Lemma~\ref{Walsh-diff} and applying Lemma \ref{lemmaPowerSimplification},   we get
\begin{align*}
  &\sum_{a,b\in \F_{2^n}}\left(\vwhtp{F}{a}{b}{4}-\vwhtp{F'}{a}{b}{4}\right)\\
  &=\sum_{a,b\in \F_{2^n}} \left(4\vwhtp{F}{a}{b}{3} E_F(a,b)-6 \vwhtp{F}{a}{b}{2}
  E_F^2(a,b)+4\vwht{F}{a}{b} E_F^3(a,b)-E_F^4(a,b) \right)\\
  &=\sum_{a,b\in \F_{2^n}} \left(4\vwhtp{F}{a}{b}{3} E_F(a,b)-12 \vwhtp{F}{a}{b}{2} D_F(b)+16\vwht{F}{a}{b} E_F(a,b)-8D_F(b)\right).
\end{align*}
Thus,
\begin{align*}
  &\frac14\sum_{a,b\in \F_{2^n}}\left( \vwhtp{F}{a}{b}{4}-\vwhtp{F'}{a}{b}{4}\right) \\
  &     =\sum_{a,b\in \F_{2^n}} \left(\vwhtp{F}{a}{b}{3} E_F(a,b)-3 \vwhtp{F}{a}{b}{2} D_F(b)\newline +4\vwht{F}{a}{b} E_F(a,b)-2D_F(b)\right).
\end{align*}
We now observe that
$\displaystyle
\sum_{a,b\in\F_{2^n}} D_F(b)=2^n\sum_{b\in\F_{2^n}} D_F(b)=2^n
\sum_{b\in\F_{2^n}} (1-(-1)^{\Tr(b\, \epsilon)})=2^{2n},
$
since  $\sum_{b\in\F_{2^n}} (-1)^{\Tr(b\, \epsilon)}=0$ when
$\epsilon\neq 0$. Further, by Parseval's identity we get
$\displaystyle
\sum_{a,b\in\F_{2^n}} \vwhtp{F}{a}{b}{2} D_F(b)=\sum_{b\in\F_{2^n}} D_F(b) \sum_{a\in\F_{2^n}} \vwhtp{F}{a}{b}{2}
=2^{2n}\sum_{b\in\F_{2^n}} D_F(b) =2^{3n}.
$
Finally, we use the inverse Walsh-Hadamard transform to obtain
\allowdisplaybreaks[4]
\begin{align*}
 & \sum_{a,b\in\F_{2^n}} \vwht{F}{a}{b} E_F(a,b) =\sum_{a,b,u\in\F_{2^n}}  (-1)^{\Tr(b(F(u)+y_0)+a(u+x_0))} D_F(b)\\
&\qquad\qquad=\sum_{b,u\in\F_{2^n}} \left(D_F(b) (-1)^{\Tr(b(F(u)+y_0))}
 \sum_{a\in\F_{2^n}}  (-1)^{\Tr(a(u+x_0))}\right)\\
&\qquad\qquad=2^n\sum_{b\in\F_{2^n}} \left(    D_F(b)  (-1)^{\Tr(b(F(x_0)+y_0))} \right)=2^n\sum_{b\in\F_{2^n}}D_F(b)=2^{2n}.
\end{align*}

Combining the above results, we obtain
\begin{align*}
\frac14\sum_{a,b\in \F_{2^n}}\left(
\vwhtp{F}{a}{b}{4}-\vwhtp{F'}{a}{b}{4}\right) =\sum_{a,b\in \F_{2^n}}
\vwhtp{F}{a}{b}{3} E_F(a,b)-(3\cdot 2^{3n}-2^{2n+1}),
\end{align*}
and our first claim is shown.

By a similar argument as in   part $(i)$,  we obtain
\allowdisplaybreaks
\begin{align}
  &\sum_{a,b\in \F_{2^n}}\left(\vwhtp{F}{a}{b}{3}-\vwhtp{F'}{a}{b}{3}\right)\notag\\
  &=\sum_{a,b\in \F_{2^n}}\left(3\vwhtp{F}{a}{b}{2} E_F(a,b)-3\vwht{F}{a}{b} E_F^2(a,b)+E_F^3(a,b) \right)\label{eq:Wcube}\\
  &=3\sum_{a,b\in \F_{2^n}} \vwhtp{F}{a}{b}{2} E_F(a,b)-6\sum_{a,b\in \F_{2^n}}\vwht{F}{a}{b} D_F(b)+4\sum_{a,b\in \F_{2^n}}E_F(a,b).\notag
\end{align}
Furthermore, with $\epsilon = y_1 - y_0$, we compute
\allowdisplaybreaks
\begin{align*}
  &\sum_{a,b\in \F_{2^n}}\vwht{F}{a}{b} D_F(b)\\
&=\sum_{b\in \F_{2^n}} \left(1-(-1)^{\Tr(b\epsilon)}\right) \sum_{u\in \F_{2^n}} (-1)^{\Tr(bF(u))} \sum_{a\in \F_{2^n}}  (-1)^{\Tr(au)}\\
&= 2^n\sum_{b\in \F_{2^n}} \left(1-(-1)^{\Tr(b\epsilon)}\right)   (-1)^{\Tr(bF(0))} \\
&=2^n \left(\sum_{b\in \F_{2^n}} (-1)^{\Tr(bF(0))} - \sum_{b\in \F_{2^n}} (-1)^{\Tr(b(y_1-y_0+F(0))}\right)\\
&=2^{2n} \left(\delta_0(F(0))- \delta_0(y_1-y_0+F(0))\right),
\end{align*}
and
\allowdisplaybreaks[4]
\begin{align*}
\sum_{a,b\in \F_{2^n}} E_F(a,b)
 &= \sum_{a,b\in \F_{2^n}}(-1)^{\Tr(ax_0+by_0)} \left(1-(-1)^{\Tr(b(y_1-y_0))} \right)\\
 &= 2^{2n}\delta_0(x_0) \left(\delta_0(y_0)-\delta_0(y_1) \right).
\end{align*}

Using these identities in~\eqref{eq:Wcube}, we obtain
\allowdisplaybreaks[4]
\begin{align*}
  &\sum_{a,b\in \F_{2^n}}\left(\vwhtp{F}{a}{b}{3}-\vwhtp{F'}{a}{b}{3}\right)\\
  &=3\sum_{a,b\in \F_{2^n}}  \vwhtp{F}{a}{b}{2} E_F(a,b)-3\cdot 2^{2n+1} \left(\delta_0(F(0))- \delta_0(y_1-y_0+F(0))\right)\\
&\qquad\qquad\qquad\qquad\qquad\qquad\qquad\qquad\quad +2^{2n+2} \delta_0(x_0) \left(\delta_0(y_0)-\delta_0(y_1) \right),
\end{align*}
and the theorem is shown.
\end{proof}

\begin{corollary}
  Let  $F$ be an $(n,n)$-function satisfying $F(0)=0$, and $x_0\in\F_{2^n}$, $\epsilon\in\F_{2^n}^*$. Let further $F'$ be its $(x_0,\epsilon)$-modification. Then we have, with $y_1 = F(x_0) + \epsilon:$
\begin{enumerate}[$(a)$]
\item
 if $x_0=y_0=0$ then $y_1\neq0$ and
 $$\sum_{a,b\in \F_{2^n}}\left(\vwhtp{F}{a}{b}{3}-\vwhtp{F'}{a}{b}{3}\right)=3\sum_{a,b\in \F_{2^n}}  \vwhtp{F}{a}{b}{2} E_F(a,b)-2^{2n+1};$$
\item if $x_0\neq 0$ then
  $$\sum_{a,b\in \F_{2^n}}\left(\vwhtp{F}{a}{b}{3}-\vwhtp{F'}{a}{b}{3}\right)=3\sum_{a,b\in \F_{2^n}}  \vwhtp{F}{a}{b}{2} E_F(a,b)-3\cdot 2^{2n+1}.$$
\end{enumerate}
\end{corollary}
\begin{proof} Follows easily from Theorem~\ref{4thpower}~$(ii)$.
\end{proof}

\begin{corollary}
\label{CorCond}
 Let $F$ be an APN $(n,n)$-function satisfying $F(0)=0$. Let $x_0=0=y_0$, and let $F'$ be the $(0,\epsilon)$-modification of $F$. Then, $F'$ is APN if and only if $$\sum_{a,b\in \F_{2^n}}\vwhtp{F}{a}{b}{3}(-1)^{\Tr(b\epsilon)}=0.$$
\end{corollary}
\begin{proof}
 By Lemma \ref{APN:char1},
$\displaystyle
\sum_{a,b\in \F_{2^n}} \vwhtp{F}{a}{b}{4} = 2^{3n+1}(3\cdot 2^{n-1}-1),$
and $\displaystyle \sum_{a,b\in \F_{2^n}} \vwhtp{F}{a}{b}{3}=2^{2n+1}(3\cdot 2^{n-1}-1)$. Also, by the same lemma, $F'$ is APN if and only if $\displaystyle
\sum_{a,b\in \F_{2^n}} \vwhtp{F'}{a}{b}{4} = 2^{3n+1}(3\cdot 2^{n-1}-1)
$. This, together with Theorem~\ref{4thpower}\,$(i)$, implies that $F'$ is APN if and only if
$$\sum_{a,b\in \F_{2^n}}\vwhtp{F}{a}{b}{3} E_F(a,b)=3\cdot 2^{3n}-2^{2n+1}.$$
On the other hand, since $x_0=0=y_0$, $E_F(a,b)=D_F(a,b)=1-(-1)^{\Tr(b\epsilon)}$, and
\begin{align*}
  &\sum_{a,b\in \F_{2^n}}\vwhtp{F}{a}{b}{3} E_F(a,b)=\sum_{a,b\in \F_{2^n}}\vwhtp{F}{a}{b}{3}-\sum_{a,b\in \F_{2^n}}\vwhtp{F}{a}{b}{3}(-1)^{\Tr(b\epsilon)}\\
  &=2^{2n+1}(3\cdot 2^{n-1}-1)-\sum_{a,b\in \F_{2^n}}\vwhtp{F}{a}{b}{3}(-1)^{\Tr(b\epsilon)}.
  \end{align*}
This, together with the previous corollary, gives the sufficient and necessary condition $\sum_{a,b\in \F_{2^n}}\vwhtp{F}{a}{b}{3}(-1)^{\Tr(b\epsilon)}=0$.
\end{proof}

\subsection{A local-global principle of APNess}

In this subsection we will show that a single point modification of an APN function is APN if and only if is partially APN.

\begin{theorem}\label{mainthm}
  Let $F$ be an $(n,n)$-function and $F'$ be its $(x_0,\epsilon)$-modification, $y_1=y_0+\epsilon$. For any $x,y \in \F_{2^n}$, let
\begin{eqnarray*}
T_{x,y}&=&\{(u,v)\in \F_{2^n}^2 : (u+x)(v+x)(u+v)\neq 0, \\
&&\qquad\qquad\qquad\qquad  F(u)+F(v)+F(u+v+x)+y=0 \},\\
S_{x,y}&=&\{u\in\F_{2^n} : F(u)+F(u+x)+y=0 \}.
\end{eqnarray*}
Then$:$
\begin{enumerate}[$(i)$]
\item[$(i)$]
$\displaystyle
\sum_{a,b\in \F_{2^n}}  \vwhtp{F}{a}{b}{3} E_F(a,b)=2^{2n}\left(3\cdot
2^n-2+\cardinality{T_{x_0,y_0}}-\cardinality{T_{x_0,y_1}} \right);
$
\item[$(ii)$]
$\displaystyle
\sum_{a,b\in \F_{2^n}}  \vwhtp{F}{a}{b}{2} E_F(a,b)=2^{2n}\left(\cardinality{S_{x_0,y_0}}-\cardinality{S_{x_0,y_1}} \right).
$
\end{enumerate}
\end{theorem}

\begin{proof}
To show $(i)$, we write
\allowdisplaybreaks
\begin{align}
  &\sum_{a,b\in \F_{2^n}}  \vwhtp{F}{a}{b}{3} E_F(a,b)
=\sum_{a,b\in \F_{2^n}}  \left(1-(-1)^{\Tr(b\, \epsilon)}\right) \notag \\
&\qquad\qquad\quad\cdot \sum_{u,v,w\in \F_{2^n}}
(-1)^{\Tr(b(F(u)+F(v)+F(w)+y_0))}
(-1)^{\Tr(a(u+v+w+x_0))} \notag\\
&=\sum_{b,u,v,w\in \F_{2^n}}  \left(1-(-1)^{\Tr(b\, \epsilon)}\right) (-1)^{\Tr(b(F(u)+F(v)+F(w)+y_0))} \notag\\
&\qquad\qquad\quad\cdot
\sum_{a\in \F_{2^n}} (-1)^{\Tr(a(u+v+w+x_0))} \notag\\
&=  \ 2^n
\sum_{b,u,v\in \F_{2^n}}  \left(1-(-1)^{\Tr(b\, \epsilon)}\right) (-1)^{\Tr(b(F(u)+F(v)+F(u+v+x_0)+y_0))} \notag\\
&=2^n \sum_{u,v\in \F_{2^n}}  \left( \sum_{b\in \F_{2^n}} (-1)^{\Tr(b(F(u)+F(v)+F(u+v+x_0)+y_0))}\right. \label{keyeq1}\\
&\left.\qquad\qquad\quad - \sum_{b\in \F_{2^n}}
(-1)^{\Tr(b(F(u)+F(v)+F(u+v+x_0)+y_1))} \right). \label{keyeq2}
\end{align}


Now, the inner sums in \eqref{keyeq1} and \eqref{keyeq2} will be zero unless one of the exponents is zero, that is, unless $F(u)+F(v)+F(u+v+x_0)+F(x_0)=0$ or $F(u)+F(v)+F(u+v+x_0)+y_1=0$.

Since there are $3\cdot 2^{n}-2$ pairs $(u,v)$ satisfying
$(u+x_0)(v+x_0)(u+v)=0$, the above equation becomes
\begin{align*}
  \sum_{a,b\in \F_{2^n}}  \vwhtp{F}{a}{b}{3} E_F(a,b) = 2^{2n}\left(3\cdot 2^n-2+\cardinality{T_{x_0,y_0}}-\cardinality{T_{x_0,y_1}} \right),
\end{align*}
and the first claim is proven.
To show $(ii)$  we write
\begin{align*}
&\sum_{a,b\in \F_{2^n}} \cW_F^2 (a,b) E_F(a,b)=\sum_{a,b\in \F_{2^n}} \left( \sum_{u,v\in \F_{2^n}} (-1)^{\Tr(a(u+v+x_0)+b(F(u)+F(v)+y_0))}\right.\\
&\left.\qquad\qquad\qquad\qquad\qquad\qquad\qquad-\sum_{u,v\in \F_{2^n}} (-1)^{\Tr(a(u+v+x_0)+b(F(u)+F(v)+y_1))}\right)\\
&=\sum_{b\in \F_{2^n}} \sum_{u,v\in\F_{2^n}}  \left((-1)^{\Tr(b(F(u)+F(v)+y_0))}-(-1)^{\Tr(b(F(u)+F(v)+y_1))}\right) \\
&\qquad\qquad\qquad\qquad\qquad\qquad\qquad\qquad\qquad\qquad\cdot \sum_{a\in \F_{2^n}}  (-1)^{\Tr(a(u+v+x_0))}\\
&=2^n\sum_{u\in \F_{2^n}} \sum_{b\in\F_{2^n}} \left((-1)^{\Tr(b(F(u)+F(u+x_0)+y_0))}-(-1)^{\Tr(b(F(u)+F(u+x_0)+y_1))}\right) \\
&=2^{2n}\left(|S_{x_0,y_0}|-|S_{x_0,y_1}| \right),
\end{align*}
and the theorem is proven.
\end{proof}

Note that in the above theorem we in fact showed that
\begin{align*}
&\sum_{a,b\in \F_{2^n}}  \vwhtp{F}{a}{b}{3}(-1)^{\Tr(ax_0+by_0)}
=2^{2n}\left(3\cdot 2^n-2+\cardinality{T_{x_0,y_0}} \right),\\
&\sum_{a,b\in \F_{2^n}}
\vwhtp{F}{a}{b}{3}(-1)^{\Tr(ax_0+by_1)}=2^{2n}\left(\cardinality{T_{x_0,y_1}}
\right).
\end{align*}
That is, for an $(n,n)$-function $F$ and its one point
modification $F'$ at $x_0$, Theorem~\ref{mainthm} gives
\begin{align}
  &\sum_{a,b\in \F_{2^n}}  \vwhtp{F}{a}{b}{3} E_F(a,b) \notag \\
 &\quad =\sum_{a,b\in \F_{2^n}} \vwhtp{F}{a}{b}{3}(-1)^{\Tr(ax_0+by_0)}-\sum_{a,b\in \F_{2^n}}
\vwhtp{F}{a}{b}{3}(-1)^{\Tr(ax_0+by_1)} \notag\\
&\quad =2^{2n}\left(3\cdot 2^n-2+\cardinality{T_{x_0,y_0}}\right)
-2^{2n}\left(\cardinality{T_{x_0,y_1}} \right) \label{Teq}.
\end{align}
By Theorem~\ref{4thpower}, we get
\begin{align*}
\frac14\sum_{a,b\in \F_{2^n}}\left( \vwhtp{F}{a}{b}{4}-\vwhtp{F'}{a}{b}{4}\right) & = \sum_{a,b\in \F_{2^n}}  \vwhtp{F}{a}{b}{3} E_F(a,b)-2^{2n}(3\cdot 2^n-2)\\ &=2^{2n}(\cardinality{T_{x_0,y_0}}-\cardinality{T_{x_0,y_1}}),
\end{align*}
where the last equality comes from the equation~\eqref{Teq}.

Therefore, we obtain the following equivalence:
\begin{align}\label{localglobal}
  \sum_{a,b\in \F_{2^n}}\left( \vwhtp{F}{a}{b}{4}-\vwhtp{F'}{a}{b}{4}\right)=0
\Longleftrightarrow  \cardinality{T_{x_0,y_0}}=\cardinality{T_{x_0,y_1}}.
\end{align}

The definition of $x_0$-APN implies that $F'$ is
$x_0$-APN if and only if
$
(u+x_0)(v+x_0)(u+v)\neq 0 \Longrightarrow
F'(u)+F'(v)+y_1+F'(u+v+x_0)\neq 0.
$
 However, when $(u+x_0)(v+x_0)(u+v)\neq 0$, one has
$
F'(u)+F'(v)+y_1+F'(u+v+x_0)=F(u)+F(v)+y_1+F(u+v+x_0).
$
Therefore, $F'$ is $x_0$-APN if and only if
$
(u+x_0)(v+x_0)(u+v)\neq 0 \Longrightarrow
F(u)+F(v)+y_1+F(u+v+x_0)\neq 0.
$
In other words, $F'$ is $x_0$-APN if and only if $T_{x_0,y_1}$ is the empty set.

 Now, the set $T_{x_0,y_0}$ with $y_0=F(x_0)$ is empty if and only if $F$ is $x_0$-APN.  By \eqref{localglobal} and Lemma \ref{APN:char1} we have:
\begin{theorem}\label{th2.6}
  If $F$ is APN and its $(x_0,\epsilon)$-modification $F'$ with $\epsilon\neq 0$ is $x_0$-APN, then $F'$ is APN.
\end{theorem}
Note that this can also be directly derived from the definition of one point modification. Indeed, suppose to the contrary, that $F'$ is $x_0$-APN but it is not APN. Then for some $a\ne 0$ and some $b$ the equation $F'(x+a)+F'(x)=b$ has more than 2 solutions. Let $x_1,x_2,x_3$ be three distinct solutions to this equation. We consider two cases. If $\{x_1,x_2,x_3\}\cap\{x_0,x_0+a\}=\emptyset$ then $F'(x_i+a)+F'(x_i)=F(x_i+a)+F(x_i)$ for $i\in\{1,2,3\}$ and this contradicts $F$ being APN. If $\{x_1,x_2,x_3\}\cap\{x_0,x_0+a\}\neq\emptyset$, then it contradicts the fact that $F'$ is $x_0$-APN.

In light of Theorem \ref{th2.6}, it follows that the conjecture from \cite{CB18} can be strengthened as follows:
\begin{conjecture}
\em  An $(x_0,\epsilon)$-modification of an APN function with $\epsilon \ne 0$ is  not $x_0$-APN.
\end{conjecture}
One way of showing that this is true would be to show
 $\{F(x_0)+F(u)+F(v)+F(x_0+u+v) : u,v \in \F_{2^n}\}= \F_{2^n}.$
 Indeed, suppose that $F'$ is an $(x_0,\epsilon)$-modification of $F$ with $y_1=y_0+\epsilon \ne y_0 = F(x_0)$ and that $F'$ is not APN. This is true if and only if the equation
 $  F'(x_0) + F'(u) + F'(v) + F'(x_0+u+v) = 0 $
 is satisfied by a pair of elements $u,v \in \F_{2^n}$ with $(u+x_0)(v+x_0)(u+v) \ne 0$. Writing $\epsilon = y_0 + y_1$, this is equivalent to
$F(x_0) + F(u) + F(v) + F(x_0+u+v) = \epsilon $
 or, in other words,  $\epsilon \in \{ F(x_0) + F(u) + F(v) + F(x_0 + u + v) : u,v \in \F_{2^n} \}$. Thus, the difference $\epsilon$ between $F(x_0)$ and $F'(x_0)$ must not be expressible as $D_aF(x_0) + D_aF(y)$ in order for $F'$ to be $x_0$-APN.

 \begin{corollary} \label{WcubeE} Let $F$ be an $(n,n)$-function and let $F'$ be its $(x_0,\epsilon)$-modification for $x_0,y_0 \in \F_{2^n}$ with $\epsilon\neq 0$. Then,
\begin{align*}
  &\sum_{a,b\in \F_{2^n}}\left(\vwhtp{F}{a}{b}{3}-\vwhtp{F'}{a}{b}{3}\right)=3\cdot  2^{2n}\left(\cardinality{S_{x_0,y_0}}-\cardinality{S_{x_0,y_1}} \right)\\
 &-3\cdot 2^{2n+1} \left(\delta_0(F(0))- \delta_0(y_1-y_0+F(0))\right)+2^{2n+2} \delta_0(x_0) \left(\delta_0(y_0)-\delta_0(y_1) \right).
\end{align*}
Furthermore,
\begin{enumerate}
\item[$(a)$] If $F(0)=0\neq x_0$, then,\\
  $\displaystyle \sum_{a,b\in \F_{2^n}}\left(\vwhtp{F}{a}{b}{3}-\vwhtp{F'}{a}{b}{3}\right)
=3\cdot  2^{2n}\left(\cardinality{S_{x_0,y_0}}-\cardinality{S_{x_0,y_1}} \right)-3\cdot 2^{2n+1};$
\item[$(b)$] If $F(0)=0=x_0$, then\\
$\begin{array}{ll}\displaystyle
\sum_{a,b\in \F_{2^n}}\left(\vwhtp{F}{a}{b}{3}-\vwhtp{F'}{a}{b}{3}\right)&=3\cdot  2^{2n}\left(\cardinality{S_{x_0,y_0}}-\cardinality{S_{x_0,y_1}} \right)-2^{2n+1}\\
&=2^{2n+1}(3\cdot2^{n-1}-1);\end{array}
$
\item[$(c)$]
If $F$ is APN and  $F(0)=0\neq x_0$, then\\
$\allowdisplaybreaks\sum_{a,b\in\F_{2^n}}\vwhtp{F'}{a}{b}{3}=2^{2n+1}(3\cdot2^{n-1}-1)+3\cdot  2^{2n}\cardinality{S_{x_0,y_1}};$
\item[$(d)$] If $F$ is APN and $F(0)=0=x_0$, then
  $\displaystyle \sum_{a,b\in\F_{2^n}}\vwhtp{F'}{a}{b}{3}=0.$
\end{enumerate}
\end{corollary}
\begin{proof}
The main claim, item $(a)$ and the first equation in $(b)$ follow easily from Theorem~\ref{4thpower}~$(ii)$ and Theorem~\ref{mainthm}~$(ii)$. For the second equation of $(b)$, we suppose $F(0)=0= x_0$. Then, $S_{x_0,y_0}=\{u\in\F_{2^n}|F(u)+F(u)+F(0)=0 \}=\F_{2^n}$, so $\cardinality{S_{x_0,y_0}}=2^n$. Also, $S_{x_0,y_1}=\{u\in\F_{2^n}|F(u)+F(u)+F'(0)=0 \}=\emptyset$, so $\cardinality{S_{x_0,y_1}}=0$.

  To show $(c)$, we assume that $F$ is APN with $F(0)=0\neq x_0$.  Then, $S_{x_0,y_0}=\{u\in\F_{2^n}|F(u)+F(u+x_0)+F(x_0)=0 \}=\{0,x_0\}$. By Lemma~\ref{APN:char1}, we get $\sum_{a,b\in\F_{2^n}}\vwhtp{F}{a}{b}{3}=2^{2n+1}(3\cdot2^{n-1}-1)$. From this and the main claim of this corollary, we have
\[
\displaystyle 2^{2n+1}(3\cdot2^{n-1}-1)-\sum_{a,b\in\F_{2^n}}\vwhtp{F'}{a}{b}{3}=3\cdot  2^{2n}\left(2-\cardinality{S_{x_0,y_1}} \right)-3\cdot 2^{2n+1},
\]
and so,
\[
\displaystyle \sum_{a,b\in\F_{2^n}}\vwhtp{F'}{a}{b}{3}=2^{2n+1}(3\cdot2^{n-1}-1)+3\cdot  2^{2n}\cardinality{S_{x_0,y_1}}.
\]

To show $(d)$, we now suppose that $F$ is APN and $F(0)=0=x_0$. Then, by Lemma~\ref{APN:char1} and point $(b)$ of this corollary,
\begin{align*}
\displaystyle
\sum_{a,b\in\F_{2^n}}\left(\vwhtp{F}{a}{b}{3}-\vwhtp{F'}{a}{b}{3}\right)
&=2^{2n+1}(3\cdot2^{n-1}-1)-\sum_{a,b\in\F_{2^n}}\vwhtp{F'}{a}{b}{3}\\
&=2^{2n+1}(3\cdot2^{n-1}-1),
\end{align*}
which implies that
$\displaystyle \sum_{a,b\in\F_{2^n}}\vwhtp{F'}{a}{b}{3}=0,$
and the claim is shown.
\end{proof}

Note that Corollary \ref{CorCond} can also be deduced from Theorem \ref{mainthm}. Furthermore, we can deduce the following corollary:
\begin{corollary} Let $F$ be an $(n,n)$-function. Let $x_0=0=y_0$, and $F'$ be the $(0,\epsilon)$-modification of $F$ for some $\epsilon \in \F_{2^n}^*$. Then,  $\sum_{a,b\in \F_{2^n}} \cW_F^2(a,b)(-1)^{\Tr(b\epsilon)}=0$.
\end{corollary}
\begin{proof}
Using the notation of Theorem~\ref{mainthm}, $S_{0,0}=\F_2^n$, while $S_{0,\epsilon}=\emptyset$. Then, by Theorem~\ref{mainthm}, $\sum_{a,b\in \F_{2^n}} \cW_F^2(a,b)E_F(a,b)=2^{3n}$.
On the other hand, \begin{align*}
\sum_{a,b\in \F_{2^n}} \cW_F^2(a,b)E_F(b)&=\sum_{a,b\in \F_{2^n}} \cW_F^2(a,b)-\sum_{a,b\in \F_{2^n}} \cW_F^2(a,b)(-1)^{\Tr(b\epsilon)}\\
&=2^{3n}-\sum_{a,b\in \F_{2^n}} \cW_F^2(a,b)(-1)^{\Tr(b\epsilon)},
\end{align*}
which shows the corollary.
\end{proof}


\section{A characterization of partial APN functions}
\label{sec3}

We now provide a necessary and sufficient condition for a function to be $x_0$-APN. As a consequence of our theorem we can obtain the APN conditions of  Lemma~\ref{APN:char1}.

 \begin{theorem}
 Let $F$ be an $(n,n)$-function and $x_0\in\F_{2^n}$. Then $F$ is $x_0$-APN if and only if
 \[\displaystyle
\sum_{a,b\in \F_{2^n}}
\vwhtp{F}{a}{b}{3}(-1)^{\Tr(ax_0+bF(x_0))}=2^{2n+1}(3\cdot 2^{n-1}-1).
\]
 \end{theorem}
 \begin{proof}
 We have
 \allowdisplaybreaks
 \begin{align*}
   &\sum_{a,b\in \F_{2^n}} \vwhtp{F}{a}{b}{3}(-1)^{\Tr(ax_0+bF(x_0))}\\
 &=  \sum_{a,b\in\F_{2^n}} (-1)^{\Tr(ax_0+bF(x_0))} \sum_{u,v,w \in \F_{2^n}} (-1)^{\Tr(b(F(u)+F(v)+F(w))+a(u+v+w))}\\
 &=   \sum_{b,u,v,w \in \F_{2^n}} (-1)^{\Tr(b(F(u)+F(v)+F(w)+F(x_0)))} \sum_{a\in\F_{2^n}}(-1)^{\Tr(a(u+v+w+x_0))} \\
 &=2^n  \sum_{b,u,v \in \F_{2^n}} (-1)^{\Tr(b(F(u)+F(v)+F(x_0)+F(u+v+x_0)))} \\
 &= 2^n \sum_{u,v \in \F_{2^n}} \sum_{b\in \F_{2^n}}(-1)^{\Tr(b(F(u)+F(v)+F(x_0)+F(u+v+x_0)))} \\
 &= 2^{2n} \cardinality{ \{ (u,v) \in \F_{2^n}^2 : F(u) + F(v) + F(x_0) + F(u+v+x_0) = 0 \} } \\
 &= 2^{2n} \left( 3 \cdot 2^n - 2 + \cardinality{T_{x_0, y_0}} \right).
 \end{align*}
Since $T_{x_0, y_0}$ is empty if and only if $F$ is $x_0$-APN, the claim follows.
 \end{proof}


 \section{Monomial partial APN functions}
 \label{sec4}

For a monomial $F(x)=x^m$, the polynomial $G(x,y,z)=F(x)+F(y)+F(z)+F(x+y+z)$ is a
symmetric homogeneous polynomial of degree $m$, and so, $G(kx,ky,kz)=k^mG(x,y,z)$ for all $k\in \F_{2^n}$. Using this property, we show that a monomial $F$ is APN if and only if $F$
is partial APN on a subspace of dimension $1$ (that is, it is partial APN at $0$ and some $x_0\neq 0$).

\begin{proposition}\label{Prop-power}
Let $F(x)=x^m$ over $\mathbb{F}_{2^n}$. Then:
\begin{enumerate}[$(i)$]
\item If $x_0\neq0$, then $F$ is $x_0$-APN if and only if $F$ is $x_1$-APN for all $x_1 \in \F_{2^n}^*$;
 \item  $F$ is APN if and
   only if $F$ is $0$-APN and $x_1$-APN for some $x_1 \in \F_{2^n}^*$.
\end{enumerate}
\end{proposition}

\begin{proof}
Certainly, $(ii)$ is a consequence of $(i)$. To show the first claim, it will be enough to show the necessity part only.
Now, we assume that $F$ is $x_0$-APN, that is  $G(x_0, y, z) \neq 0$
for all $y, z$ with $(y+x_0)(z+x_0)(y+z) \neq 0$, and we want to show that $F$ is $x_1$-APN for any other $x_1\in \F_{2^n}^*$. By absurd, we assume that there is some $x_1\neq 0$, for which $F$ is not $x_1$-APN. Then, there exist $x_1\neq y_1\neq z_1\neq x_1$ such that $G(x_1, y_1, z_1)=0$.  Using the
homogeneous property of $G$, namely $0=k^mG(x_1, y_1, z_1) =G(kx_1, ky_1, kz_1)$
 for any $k\neq 0$, and taking $k=x_0/x_1\neq 0$, then the condition can be written as  $G(x0, y, z) = 0$ for $y=ky_1,z=kz_1$ and $y, z$ with $(y+x_0)(z+x_0)(y+z) \neq
 0$, and that is a contradiction.
\end{proof}

A partial APN concept on $(n,n)$-functions is also considered in \cite{CK17}:  $F$ is said to satisfy the property $(p_{a})$, $a\in\F_{2^n}^*$, if the equation $F(x)+F(x+a)=b$ has either 0 or 2 solutions for every $b\in\F_{2^n}$. They showed that a mapping $F$ is APN if and only if $F$ satisfies $(p_{a})$ for all nonzero $a$ belonging to a hyperplane. It is not clear if such a result is true in general for our notion of partial APNness. From the result above, we see that a similar result is true for monomials, i.e. $F$ is APN if and only if it is partial APN for a subspace of dimension $1$.
Moreover, when $F$ is a monomial, the property $(p_1)$ implies the property $(p_a)$ for
any $a\neq 0$. Therefore our result on $0$-APN has some analogy with the
 property $(p_1)$, but $0$-APN is a more general condition than the
  property $(p_1)$, as the following examples will show.

We let $\binom{a}{b}_2$ denote the residue modulo $2$ of the binomial coefficient $\binom{a}{b}$.
We next investigate and explicitly construct many classes of Boolean functions that are $0$-APN (but not necessarily APN).

\begin{theorem}
\label{xm-0APN}
Let $\F_{2^n}$ be the extension field of $\F_2$ corresponding to the primitive polynomial $f$ of degree $n$ and let $g$ be one of the (primitive) roots of $f$. Then:
\begin{enumerate}[$(i)$]
  \item[$(i)$] if $F(x)=x^m$ over $\F_{2^n}$, then $F$ is $0$-APN if and only if for $1\leq i\leq 2^n-1$, the minimal polynomial $P_{g^i}(x)=\prod_{j\in C_i} (x-g^j)$ of $g^i$, where $C_i=\{(i\cdot 2^j)\pmod{2^n-1} : j=0,1,\ldots \}$ is the unique cyclotomic coset of $i$ modulo $2^n-1$, does not divide $\sum_{k=1}^{mi-1} \binom{mi}{k}_2\ x^{mi-k-1}$;
\item[$(ii)$] if $F(x)=x^{2^d-1}$ over $\F_{2^n}$, then $F$ is  $0$-APN if and only if $\gcd(d-1, n)=1$;
\item[$(iii)$] if $F(x)=x^{2^d+1}$ (Gold exponent) over $\F_{2^n}$,  then $F$ is $0$-APN if and only if $\gcd(d, n)=1$.
\end{enumerate}
\end{theorem}

\begin{proof}
If $F(x)=x^m$, then $F$ is $0$-APN if and only if the Rodier equation
\[
F(y)+F(z)+F(y+z)=y^m+z^m+(y+z)^m=0,
\]
has no solution $y,z \in \F_{2^n}^*$ with $y \ne z$. Given two distinct elements $y,z \in \F_{2^n}^*$, let $z=y\alpha$, where $\alpha\neq 0,1$.
Then, the equation above becomes
\[
y^m \left(1+\alpha^m  +(1+\alpha)^m\right)=0,
\]
implying $
1+\alpha^m  +(1+\alpha)^m=0.
$
 Then, if there exists $\alpha\neq 0,1$ satisfying the previous equation,
 then there exists $1\leq i\leq 2^n-1$ such that
 \[
	 \displaystyle  \frac{1+x^{im}+(1+x^i)^{m}}{x}=\sum_{k=1}^{m-1} \binom{m}{k}_2\ x^{i(m-k)-1}
  \]
  vanishes at $g$, that is, $1+g^{i\,m}+\left(1+g^{i} \right)^m=0$. Then it will vanish at $g^{2^\ell}$, for all $\ell$, since $1+g^{i\,m\, 2^\ell}+\left(1+g^{i\, 2^\ell}\right)^m=\left(1+g^{i\,m}+\left(1+g^{i} \right)\right)^{2^\ell}=0$. Thus, the minimal polynomial $P_{g^i}(x)=\prod_{j\in C_i} (x-g^j)$ of $g^i$ divides $\displaystyle \sum_{k=1}^{mi-1} \binom{mi}{k}_2\ x^{mi-k-1}$. The converse is certainly true, and the first claim is shown.

To test whether $F=x^{2^d-1}$ is $0$-APN, one needs to check the (in)solvability of the Rodier equation
\begin{align*}\displaystyle
0&=F(y)+F(z)+F(y+z)\\
&=y^{2^d-1}+z^{2^d-1}+(y+z)^{2^d-1}\\
&=\frac{zy^{2^d-1}+yz^{2^d-1}}{y+z}=\frac{(\alpha^{2^d-1}+\alpha)z^{2^d}}{z(\alpha + 1)},
\end{align*}
where $y=z\alpha, \alpha\neq 0,1$. Therefore, when (and only when) $\gcd(2^d-2, 2^n-1)=1$, there is no $\alpha\neq 0,
1$ satisfying the above equation, that is, $x^{2^d-1}$ is $0$-APN. The condition $\gcd(d-1,n)=1$ follows form the known
identity $\gcd(2^a-1,2^b-1)=2^{\gcd(a,b)}-1$, since $1=\gcd(2^d-2, 2^n-1)=\gcd(2^{d-1}-1, 2^n-1)=2^{\gcd(d-1,n)}-1$.

In the same way, we consider $F(x)=x^{2^d+1}$ over $\F_{2^n}$. To test
whether $F$ is $0$-APN, one needs to check the solvability of the Rodier equation
\begin{align*}
0&=F(y)+F(z)+F(y+z)\\
&=y^{2^d+1}+z^{2^d+1}+(y+z)^{2^d+1}  \\
&=zy^{2^d}+yz^{2^d}=(\alpha^{2^d}+\alpha)z^{2^d+1},
\end{align*}
 where $y = z\alpha, \alpha \ne 0,1$. Therefore, when (and only when) $1=\gcd(2^d-1, 2^n-1)=2^{\gcd(d,n)}-1$, that is, for $\gcd(d,n)=1$, there is no $\alpha\neq 0, 1$ satisfying the above equation, so $x^{2^d+1}$ is $0$-APN.
\end{proof}

\begin{table}[htb]
  \centering
  \begin{tabular}{|c|l|c|}
    \hline
    $n$ & Exponents $i$ & $\Delta_F$ \\
    \hline
    \hline
    1-5 & - & - \\
    \hline
    6 & 27 & 12 \\
    \hline
  \multirow{2}{*}{7} & 7,21,31,55 & 6 \\
   & 19,47 & 4 \\
   \hline
   \multirow{4}{*}{8} & 15,45 & 14 \\
   & 21,111 & 4 \\
   & 51 & 50 \\
   & 63 & 6 \\
   \hline
   \multirow{3}{*}{9} & 7,21,35,61,63,83,91,111,117,119,175 & 6 \\
   & 41,187 & 8 \\
   & 45,125 & 4 \\
   \hline
   \multirow{6}{*}{10} & 15,27,45,75,111,117,147,189,207,255 & 6 \\
   & 21,69,87,237,375 & 4 \\
   & 51 & 8 \\
   & 93 & 92 \\
   & 105,351 & 10 \\
   & 231,363,495 & 42 \\
   & 447 & 12 \\
   \hline
  \end{tabular}
  \caption{Power functions $F(x) = x^i$ over $\F_{2^n}$ for $1 \le n \le 10$ that are 0-APN but not APN}
  \label{table0APNbutnot1APN}
\end{table}
\begin{example}
 Table~\textup{\ref{table0APNbutnot1APN}} lists the exponents $i$ for which $x^i$ is 0-APN but not APN over $\F_{2^n}$ for $1 \le n \le 10$. Only one representative from every cyclotomic coset is given. There are no functions of this type for  $n \le 5$.
 
While there are power functions that are partial 0-APN but not APN, this is not true for partial 1-APN power functions. The proof is, in fact, rather immediate (we thank Dr. Namhun Koo for providing the included short proof here).

 \begin{theorem}
   Any partial $1$-APN power function $F(z)=z^k$ is APN.
 \end{theorem}
 \begin{proof}
 By proposition \ref{Prop-power}, it will be sufficient to show that $f$ is $0$-APN.
 Suppose, on the contrary, that $F(z)=z^k$ is not $0$-APN. Then
there exist $x, y\in \F_{2^n}$ with $xy(x+y)\neq 0$ satisfying
$F(0)+F(x)+F(y)+F(x+y)=0$. Since $x\neq 0$,
\begin{align}
0&=F(x)+F(y)+F(x+y)=x^k+y^k+(x+y)^k \notag \\
 &=1+(y/x)^k+(1+y/x)^k=F(1)+F(y/x)+F(1+y/x) \notag \\
 &=F(1)+F(a)+F(b)+F(1+a+b), \notag
\end{align}
where
$a=\frac{y}{x}$, $b=1+\frac{y}{x}$, $1+a+b=0$.
Since $F$ is $1$-APN, one must have $(a+1)(b+1)(a+b)=0$. However,
\begin{align}
0=(a+1)(b+1)(a+b) 
 =\left(\frac{y}{x}+1\right)\cdot \frac{y}{x} \cdot 1. \notag
\end{align}
Thus we get $x=y$ or $y=0$, contradicting the fact $xy(x+y)\neq 0$.
 \end{proof}

  This is not true in general, for non-monomials: we found over six million polynomials over $\mathbb{F}_{2^3}$ that are $1$-APN but not APN,  for example, $x^7 + x^6$. Out of these, $64$ have coefficients in $\F_2$: $48$ of them have the differential spectrum $\{ 0^{31}, 2^{22}, 4^3 \}$, while the remaining $16$ have the spectrum $\{ 0^{42}, 2^{7}, 6^{7} \}$.
We also found $6944$ polynomials of this type over $\F_{2^4}$ with coefficients in $\F_2$, for example, $x^{12} + x^7$.

Nonetheless, it seems likely that if some $(n,n)$-function $F$ is $x$-APN for all $x \in \mathbb{F}_{2^n}\setminus\{x_0\}$, then it is $x_0$-APN (and hence APN) as well. This can be easily observed to be true for quadratic functions. Recall that $F$ is $x_0$-APN if for any $a \ne 0$ the equation $F(x_0) + F(x) + F(x + a) + F(x_0 + a) = D_aF(x) + D_aF(x_0) = 0$ has precisely two solutions, namely, $x = x_0$ and $x = x_0 + a$. Since $D_aF$ is an affine function, this is equivalent to $D_aF(x + x_0) = D_aF(0)$ having only $x = x_0$ and $x = a + x_0$ as solutions.

\begin{proposition}
	Let $F$ be a quadratic $(n,n)$-function and $x_0 \in \mathbb{F}_{2^n}$. Then $F$ is $x_0$-APN if and only if $F$ is APN.
\end{proposition}

\end{example}

\section{Classes of never  $0$-APN (hence never APN) for infinitely many extensions of $\F_2$}
\label{sec5}

Building up on some of their earlier work on the function $x^3+\Tr(x^9)$, which is APN on $\F_{2^n}$, for all dimensions $n$, Budaghyan et al.~\cite{BCL09} generalized this class to  $L_1(x^3)+L_2(x^9)$, where $L_1,L_2$ are linear functions  on $\F_{2^n}$, and found conditions under which this function is APN.

In a series of papers, Rodier and his collaborators~\cite{AMR10,FOR12,LR11,Rod11,Rod09}
concentrated on finding classes of functions that are never APN for infinitely many extensions of the prime field $\F_2$. Here we present classes of functions that are never $0$-APN (and hence never APN) for infinitely many extensions of $\F_2$, and in the process even extend some of the existing results.


\begin{theorem}
  Let $L$ be a linear polynomial on $\F_{2^n}$, $g$ be a primitive element of $\F_{2^n}$ and $d \ge 1$ be a positive integer. Furthermore, let $F$ and $G$ be defined over $\F_{2^n}$ by $F(x)=L\left(x^{2^d+1}\right)+\Tr(x^{3})$ and $G(x)=L\left(x^{2^{d+1}+2^d+1}\right)+\Tr(x^{3})$.
If $\gcd(d,n)>1$, then neither $F$ nor $G$ is $0$-APN.

In general, $L\left(x^m\right)+\Tr(x^3)$ is not $0$-APN if there exists some $1\leq i\leq 2^n-1$, such that $P_{g^i}(x)=\prod_{j\in C_i} (x-g^j)$ divides $\displaystyle \sum_{k=1}^{m-1} \binom{m}{k}_2\ x^{i(m-k)-1}$, where $C_i=\{(i\cdot 2^j)\pmod{2^n-1}\,|\, j=0,1,\ldots \}$ is the unique cyclotomic coset of $i$ modulo $2^n-1$.
\end{theorem}

\begin{proof}
  The function $F$ is $0$-APN if and only if there are no solutions $x,y \in \F_{2^n}^*$, $x \ne y$ of the equation
\begin{align*}
0
&=F(x)+F(y)+F(x+y)\\
&=L\left(x^{2^d+1}+y^{2^d+1}+(x+y)^{2^d+1}\right)+\Tr\left(x^3+y^3+(x+y)^3\right)\\
&=L\left(x^{2^d}y+x\,y^{2^d}\right)+\Tr(x^2y+xy^2).
\end{align*}
Writing $y = \alpha x$, this is equivalent to the equation
\begin{align*}
&L\left(x^{2^d+1}(\alpha+\alpha^{2^d})\right)= \Tr\left(x^3(\alpha+\alpha^2)\right)
\end{align*}
having no solution for $\alpha \ne 0,1$.
Now, if  $m=\gcd(d,n)>1$, we take $\alpha\in\F_{2^m}\setminus\F_2\subseteq \F_{2^d}\cap  \F_{2^n}$. Then $\alpha^{2^d}+\alpha=0$, and for $x=1$ we have $\Tr(x^3(\alpha+\alpha^2))=0$, since it is known that  $\Tr(u)=0$ if and only if $u=b^2+b$ (in characteristic $2$), which renders nontrivial solutions to the above equation.  The first claim is shown.

We now concentrate on $G(x)$. Once again we want to show that the Rodier equation
\[ G(x) + G(y) + G(x+y) = 0 \]
has no solutions $x,y \in \F_{2^n}^*$ with $x \ne y$. Similarly to the case for $F$ above and writing $y = \alpha x$, we can easily see that this is equivalent to the equation
\begin{equation}
  L\left( x^{2^{d+1} + 2^d + 1} (\alpha + \alpha^{2^d})(1 + \alpha^{2^d} + \alpha^{2^{d+1}}) \right) = \Tr \left(x^3(\alpha + \alpha^2) \right)
  \label{eqRodier3}
\end{equation}
having no solutions with $\alpha \ne 0,1$.
So, denoting $m=\gcd(d,n)>1$, we can take $\alpha\in\F_{2^m}\setminus\F_2\subseteq \F_{2^d}\cap  \F_{2^n}$. Then we have $\alpha + \alpha^{2^d} = 0$ so that this $\alpha$ along with $x=1$ constitute a solution to \eqref{eqRodier3} implying that $G$ is not $0$-APN.

The  last claim can be argued as in the proof of Theorem~\ref{xm-0APN}$(i)$.
\end{proof}
\begin{remark}
	The condition on $d$ in the above theorem is important. Indeed, we have computationally checked that if $n=5$, then $x^9+\Tr(x^3)$ is $0$-APN, and potentially there may be some other cases.
\end{remark}

These classes of functions can be further generalized so as to encompass even more functions that are not $0$-APN.
\begin{theorem}
  Let $L_1$ and $L_2$ be linear functions over $\F_{2^n}$. If $\gcd(d,r,n) > 1$, then $L_1(x^{2^d+1})+L_2(x^{2^r+1})$ is not $0$-APN.

Furthermore, if $L_1$ is the identity and $L_2$ is the absolute trace, then $x^{2^d+1}+\Tr(x^{2^r+1})$ is not $0$-APN if $\gcd(d,n)>1$ and $\gcd(2^r+1,2^n-1)=1$, or $\gcd(d,r,n)>1$.

Finally, if $\gcd(d,s,n)>1$, then $L_1\left(x^{2^{d+1}+2^d+1}\right)+L_2\left(x^{2^{s+1}+2^s+1}\right)$ is not $0$-APN.
\end{theorem}
\begin{proof}
We consider first the function $L_1(x^{2^d+1})+L_2(x^{2^r+1})$.
As before, we investigate the solvability of the equation
\begin{align}
\label{eq:L12}
L_1\left(x^{2^d+1}(\alpha+\alpha^{2^d})\right)=L_2\left(x^{2^r+1}(\alpha+\alpha^{2^r})\right),
\end{align}
where $y = \alpha x$ for $x \ne 0$ and $\alpha \ne 0,1$.
Denoting $m=\gcd(d,r,n)>1$, we can take $\alpha\in\F_{2^m} \setminus\F_2\subseteq \F_{2^d}\cap \F_{2^r}$. Then $\alpha^{2^d}+\alpha=\alpha^{2^r}+\alpha=0$, so that~\eqref{eq:L12} has nontrivial solutions and thus the considered function is not $0$-APN.

In the particular case when $L_1$ is the identity and $L_2$ is the trace function, it is sufficient to show that the function
$x^{2^d+1}+\Tr(x^{2^r+1})$ is not $0$-APN if $\gcd(d,n)>1$ and $\gcd(2^r+1,2^n-1)=1$ since the other case follows from the previously proven statement. The relevant Rodier equation is
\[
x^{2^d+1}(\alpha+\alpha^{2^d}) =\Tr\left(x^{2^r+1}(\alpha+\alpha^{2^r})\right).
\]
Denoting $m=\gcd(d,n)>1$, we can find $\alpha\in\F_{2^m} \setminus\F_2$ for which the left hand side vanishes. Now we argue that regardless of the value of $\alpha$, there exists an element $x$ such that $x^{2^r+1}(\alpha+\alpha^{2^r})=\beta^2+\beta$ for some $\beta$. If $\alpha+\alpha^{2^r}=0$, we are done since $x$ can take any value. If $\alpha+\alpha^{2^r}\neq 0$, taking $\beta=\alpha+\alpha^{2^r}$, if $\beta+1\neq 0$, or any other nonzero element $\beta$ of the finite field such that $\beta+1\neq 0$, the above claim  is implied by the existence of solutions $x$ such that $x^{2^r+1}=\frac{\beta^2+\beta}{\alpha+\alpha^{2^r}}$. This in turn follows from the fact that $\gcd(2^r + 1, 2^n-1) = 1$ and thus every element of $\F_{2^n}$ has a $2^r+1$-st root (see e.g.~\cite{ln}).
\\
To show the last claim, we again examine the relevant Rodier equation which in this case (by applying the same approach as above) takes the form
\begin{align*}
L_1\left(\left( \alpha+\alpha^{2^d}\right)\left( 1+\alpha^{2^d}+\alpha^{2^{d+1}}\right)\right)=
L_2\left(\left( \alpha+\alpha^{2^s}\right)\left( 1+\alpha^{2^s}+\alpha^{2^{s+1}}\right)\right).
\end{align*}
Denoting $m=\gcd(d,s,n)>1$, we can find $\alpha\in\F_{2^{m}}\setminus \F_2$, so that , $\alpha+\alpha^{2^d}=\alpha+\alpha^{2^s}=0$. The Rodier equation thus has nontrivial solutions and the function in question is not $0$-APN.
\end{proof}

Recall the following result (obtained using a combination of theoretical and computational arguments) of Leander and Rodier~\cite{LR11}.
\begin{theorem}[Leander-Rodier, 2011]
\label{thm:lr11}
If $n\geq 2$ and $d$ is a nonzero integer which is not a power of $2$, then the function
\[
F(x)=x^{2^n-2}+\beta\, x^d
\]
over $\F_{2^n}$ is not APN for $d\leq 29$ and any $\beta \in \F_{2^n}^*$.
\end{theorem}
Below we find more classes of functions that are not $0$-APN for infinitely many extensions $\F_{2^n}$. In the process, we extend the  previous result of Leander and Rodier.
\begin{theorem}
Let $a>b$ be positive integers. Assuming that one of $x^a$ and $x^b$ are $0$-APN on $\F_{2^n}$ and $\gcd(a-b,
2^n-1)=1$, the polynomial $x^a+\beta\,x^b$ is not $0$-APN for any $\beta\in\F_{2^n}^*$. Let $c>d$ be positive integers. In particular,
\begin{enumerate}[(i)]
\item if
$\gcd(c-1,n)=\gcd(c-d,n)=1$, or $\gcd(d-1,n)=\gcd(c-d,n)=1$, then the polynomial
$x^{2^c-1}+\beta\,x^{2^d-1}$ is not $0$-APN;
\item if $\gcd(c,n)=\gcd(c-d,n)=1$, or $\gcd(d,n)=\gcd(c-d,n)=1$, then the polynomial
$x^{2^c+1}+\beta\,x^{2^d+1}$ is not $0$-APN;
\item if $\gcd(c,n)=\gcd(2^{c-1}-2^{d-1}+1,2^n-1)=1$, or $\gcd(d-1,n)=\gcd(2^{c-1}-2^{d-1}+1,2^n-1)=1$, then the polynomial $x^{2^c+1}+\beta\,x^{2^d-1}$ is not $0$-APN;
\item if $\gcd(c-1,n)=\gcd(2^{c-1}-2^{d-1}-1,2^n-1)=1$, or $\gcd(d,n)=\gcd(2^{c-1}-2^{d-1}-1,2^n-1)=1$, then the
polynomial $x^{2^c-1}+\beta\,x^{2^d+1}$ is not $0$-APN.
\end{enumerate}
\end{theorem}
\begin{proof}
Let $F(x)=x^a+\beta\,x^b \,\, (a>b)$. Then $F$ is $0$-APN if and only if
$0=F(y)+F(z)+F(y+z)$ has no solutions $y,z$ with $yz(y+z)\neq 0$.
The relevant Rodier equation takes the form
\begin{align*}
  0&=F(y)+F(z)+F(y+z)=y^a+\beta\,y^b+z^a+\beta z^b+(y+z)^a+\beta(y+z)^b,
\end{align*}
which, with $y = z\alpha$ with $\alpha \ne 0,1$, becomes
\begin{align*}
 0 &=z^a\left(\alpha^a+1+(\alpha+1)^a\right)+ \beta z^b\left(\alpha^b+1+(\alpha+1)^b \right).
\end{align*}
Note that the polynomial $x^m$ is $0$-APN if and only
$x^m+1+(x+1)^m$ has no root $x\neq 0,1$, and such $m$ can
be classified by Theorem~\ref{xm-0APN}\,$(i)$. Assume that at least one of $x^a$
and $x^b$ are $0$-APN.
Then one can always find
$\alpha \in \F_{2^n}$ such that
$$
\alpha^a+1+(\alpha+1)^a\neq 0\neq \alpha^b+1+(\alpha+1)^b.
$$
For example, when $x^a$ is $0$-APN, one can choose any $\alpha\neq
0,1$ outside the roots of $x^b+1+(x+1)^b=0$. Therefore one has
$$
z^{a-b}=\beta\frac{\alpha^b+1+(\alpha+1)^b}{\alpha^a+1+(\alpha+1)^a}.
$$
When $\gcd(a-b,2^n-1)=1$, the above equation always has a unique
solution $z$ for any $\alpha\neq 0,1$, and one has $y =z\alpha\neq z$, since $\alpha\neq 1$.

We now show the other claims.
When $a=2^c-1$ and $b=2^d-1$, with $\gcd(c-1,n)=1$ or $\gcd(d-1,n)$,  then by Theorem \ref{xm-0APN}, one of $x^a$ or $x^b$ is 0-APN. One
has $a-b=2^d(2^{c-d}-1)$ and
$\gcd(a-b,2^n-1)=\gcd(2^{c-d}-1,2^n-1)=2^{\gcd(c-d,n)}-1$, which
becomes one if and only if $\gcd(c-d,n)=1$. Therefore, when
$\gcd(c-d,n) = 1$ the polynomial $x^{2^c-1}+\beta\,x^{2^d-1}$ is not $0$-APN by the first part of the proof.


When $a=2^c+1$ and $b=2^d+1$ with   $\gcd(c,n)=1$ or
$\gcd(d,n)=1$,  then by Theorem~\ref{xm-0APN}, one of $x^a$ or $x^b$ is 0-APN. One
has $a-b=2^d(2^{c-d}-1)$ and
$\gcd(a-b,2^n-1)=\gcd(2^{c-d}-1,2^n-1)=2^{\gcd(c-d,n)}-1$ which
becomes one if and only if $\gcd(c-d,n)=1$.  Therefore, when $\gcd(c-d,n)=1$, the polynomial $x^{2^c+1}+\beta\,x^{2^d+1}$ is not $0$-APN.

When $a=2^c+1$ and $b=2^d-1$ with $\gcd(c,n)=1$ or
$\gcd(d-1,n)=1$,  then by Theorem \ref{xm-0APN}, one of $x^a$ or $x^b$ is 0-APN. One has $a-b=2^c-2^d+2$ and
$\gcd(a-b,2^n-1)=\gcd(2^{c-1}-2^{d-1}+1,2^n-1)$.   Therefore,
when $\gcd(2^{c-1} - 2^{d-1} + 1, 2^n-1) = 1$, the polynomial $x^{2^c+1}+\beta\,x^{2^d-1}$ is not $0$-APN.

Lastly, when $a=2^c-1$ and $b=2^d+1$ with   $\gcd(c-1,n)=1$ or
$\gcd(d,n)=1$,  then by Theorem \ref{xm-0APN}, one of $x^a$ or $x^b$ is 0-APN. One has $a-b=2^c-2^d-2$ and
$\gcd(a-b,2^n-1)=\gcd(2^{c-1}-2^{d-1}-1,2^n-1)$.  Therefore,
when $\gcd(2^{c-1} - 2^{d-1} - 1, 2^n-1) = 1$ the polynomial $x^{2^c-1}+\beta\,x^{2^d+1}$ is not $0$-APN.
\end{proof}
From the above examples, one can find many binomials which are
not $0$-APN for infinitely many extensions of the prime field $\F_2$. For
example, both $x^7+x^3$ and $x^5+x^3$ are not $0$-APN for all
finite fields $\F_{2^n}$ when $n>2$.
We can easily generalize (for any odd $n$)   Leander and Rodier's result of Theorem~\ref{thm:lr11}~\cite{LR11} in our next corollary.
\begin{corollary}
  Assume that $n$ is odd and $d$ is a positive integer with $\gcd(d+1,2^n-1)=1$. Then $x^{2^n-2}+\beta\, x^d$ is not $0$-APN for any $\beta \in \F_{2^n}^*$.
\end{corollary}
\begin{proof}
Observe that $x^{2^n-2}$ is APN for $n$ odd. By the previous theorem $x^{2^n-2}+\beta\,x^d$ is not APN if $1=\gcd(2^n-2-d, 2^n-1)=\gcd(2^n-1,d+1)$ and the proof is done.
\end{proof}

\section{Conclusion and further comments}\label{sec6}

In this paper we introduce a partial APN (pAPN) concept, which may help in understanding the APN property and its properties. We certainly just scratched the surface in the investigation of the pAPN notion and there are certainly many more questions one could ask. For example, we propose further constructions of large classes of such pAPN functions, as well as perhaps look into the construction of permutation pAPN, which may shed light into the well known and quite difficult problem of the permutation APN problem.

\vskip.5cm

\noindent {\bf Acknowledgements.} 
The authors would like to thank the referees for their thorough reading and useful comments, and the editors for handling our manuscript very efficiently. The paper was started while the fourth named author visited Selmer center at UiB in the Summer of 2018. This author thanks the institution for the excellent working conditions. S.K. was supported by the National Research Foundation of Korea (NRF) grant funded by the Korea government (MSIP) (No. 2016R1D1A1B03931912 and No. 2016R1A5A1008055). The research of the first two authors was supported by Trond Mohn foundation.

\end{document}